%% file: arXivArticle.tex
\documentclass[preprint,12pt,3p]{elsarticle}

\usepackage{amssymb}
\usepackage{amsthm}
\usepackage{graphicx}
\usepackage{amsmath}
\usepackage{hyperref}

\usepackage{pgf,tikz}

\usetikzlibrary{arrows}
\usetikzlibrary{shapes}
\usepackage{multirow}
\usepackage{latexsym}
\usepackage{setspace}
\PassOptionsToPackage{boxed,section}{algorithm}
\usepackage{enumerate}
\usepackage{amsmath}			% special AMS math definitions
\usepackage{amssymb}			% special AMS math symbols
\usepackage{url}
\usepackage{setspace}
\usepackage{tikz}
\usetikzlibrary{arrows}
\usetikzlibrary{shapes}
\usetikzlibrary{decorations.markings}
\usepackage{geometry}

\usepackage{comment}

\newtheorem{proposition}{\em Proposition}
\newtheorem{theorem}{\em Theorem}

\newtheorem{definition}{\em Definition}
\newtheorem{lemma}{\em Lemma}
\newtheorem{corollary}{\em Corollary}

\newtheorem{question}{\em Question}

\usepackage{color} % e.g.\textcolor{red}{text example}
\usepackage{amsmath, amsfonts, amssymb}
\usepackage{xspace}
\usepackage{standalone}

\newcommand{\N}{\mathbb{N}}
\newcommand{\card}[1]{\vert #1 \vert}
\DeclareMathOperator*{\argmax}{arg\,max}

\journal{Sample Journal}

\begin{document}

\begin{frontmatter}

\title{On the fixed-parameter tractability of the maximum connectivity improvement problem}
%\tnotetext[label0]{This is only an example}

\author[label1]{Federico Cor\`o\corref{cor1}}
\address[label1]{Gran Sasso Science Institute, L'Aquila, Italy}

\cortext[cor1]{Corresponding author}
%\fntext[label3]{The author is supported by a grant \ldots}
%\fntext[label4]{Small city}

\ead{federico.coro@gssi.it}
%\ead[url]{author-one-homepage.com}

\author[label1]{Gianlorenzo D'Angelo}
%\address[label5]{Department of Informatics and Applied Mathematics, Yerevan State University, Yerevan, Armenia}
\ead{gianlorenzo.dangelo@gssi.it}

\author[label1]{Vahan Mkrtchyan}
%\address[label5]{Department of Informatics and Applied Mathematics, Yerevan State University, Yerevan, Armenia}
\ead{vahan.mkrtchyan@gssi.it}

%\author[label1,label5]{Author Three}
%ead{author.three@mail.com}

\begin{abstract}
In the Maximum Connectivity Improvement (MCI) problem, we are given a directed graph $G=(V,E)$ and an integer $B$ and we are asked to find $B$ new edges to be added to $G$ in order to maximize the number of connected pairs of vertices in the resulting graph. 
The MCI problem has been studied from the approximation point of view. In this paper, we approach it from the parameterized complexity perspective in the case of directed acyclic graphs. 
We show several hardness and algorithmic results with respect to different natural parameters. 
Our main result is that the problem is $W[2]$-hard for parameter $B$ and it is FPT for parameters $|V| - B$ and $\nu$, the matching number of $G$. 
We further characterize the MCI problem with respect to other complementary parameters.

\end{abstract}


\end{frontmatter}

\section{Introduction}
Graph augmentation problems under connectivity requirements are fundamental topics in algorithmic research. Given a graph, these problems ask to add a limited number of new edges to it in order to meet some connectivity requirements, like, for example, strong connectivity of a directed graph~\cite{eswaran1976augmentation}.

While graph augmentation problems have been widely investigated (see~\cite{BG08,frank2011connections}), their study from a parameterized complexity perspective is still at the beginning.
%
%
%Graph augmentation problems with connectivity requirements have been already studied in the field of parameterized complexity.
Marx and V\'egh~\cite{MV15} proved that the problem of increasing the edge-connectivity of an undirected graph from $k-1$ to $k$ by adding at most $B$ edges is fixed-parameter tractable when parameterized by $B$. 
They proposed an algorithm with complexity $2^{O(B\log B)}|V|^{O(1)}$. 
This bound on the running time was then improved by Basavaraju et al.~\cite{BFGMRS14}, that provided an algorithm that runs in $9^B |V|^{O(1)}$ time. 
Marx and V\'egh also proved the fixed-parameter tractability of increasing edge-connectivity from 0 to 2 and increasing vertex-connectivity from 1 to 2~\cite{MV15}. 

Gao et al.~\cite{GHN13}, showed that the problem of adding $B$ edges to a graph in order to have a diameter equal to $t$, for some input $t$, is $W[2]$-hard with respect to parameter $B$, for every $t$. 
Hoffmann et al.~\cite{HMS17} proved that the problem of maximizing the closeness or the betweenness centrality of a graph by adding $B$ new edges is $W[2]$-hard with respect to parameter $B$ and give FPT algorithms for a parameter that measures the distance to cluster graphs.

%Basavaraju et al.~\cite{BFGMRS14}, Bang-Jensen et al.~\cite{BBKMRS18}, and  Gutin et al.~\cite{GRRW19},
A related class of problems asks to select a set of edges in a graph in order to maintain some connectivity requirements when these edges are removed.
%is that of edge deletion in which a set of edges of a graph must be deleted, while maintaining. 
%
A general version of this problem is the so called \emph{Survivable Network Design Problem} (SND), in which we are given an edge-weighted (directed or undirected) graph, a connectivity requirement function $r:V\times V\rightarrow \mathbb{N}$, and two numbers, $B$ and $W$. 
The aim is to remove a set of edges of size at most $B$ and overall cost at least $W$ in such a way that the resulting graph has at least $r(u,v)$ edge-disjoint (or vertex-disjoint) paths for every $(u,v)\in V\times V$. 
The SND problem is $W[1]$-hard in directed graphs for both vertex-disjoint and edge-disjoint versions, even in the case of uniform weights~\cite{BBKMRS18}.
The case of \emph{uniform demands}, in which function $r$ is a constant and the aim is to preserve edge-connectivity or vertex-connectivity of $k$-connected graphs, admits FPT algorithms in several interesting cases.
Basavaraju et al.~\cite{BFGMRS14} gave a $2^{O(B)}|V|^{O(1)}$ time algorithm for the unweighted version of the edge-connectivity problem in undirected graphs.
%which $B$ edges must be removed while preserving the edge-connectivity of a $k$-edge-connected undirected graph.
%
Bang-Jensen et al.~\cite{BBKMRS18} considered the problem of maintaining the edge-connectivity or vertex-connectivity of both directed and undirected weighted graphs.
%in which edges are associated with costs and a set of edges of size at most $B$ and overall cost at least $w$ must be removed. 
They provided a $2^{O(B\log (B+k))}|V|^{O(1)}$ time algorithm for the case of edge-connectivity of both directed and undirected graphs and an algorithm with same running time for the case of vertex-connectivity in directed graphs. 
Gutin et al.~\cite{GRRW19} considered this latter problem in undirected graphs and showed that it admits a $2^{O(B\log B)}|V|^{O(1)}$ time algorithm in the case of biconnected graphs (i.e., $k=2$).

Several optimization problems related to graph augmentation have been studied in the field of approximation algorithms. 
For the problem of computing a minimum cost set of new edges to add to an existing $k$-connected graph in order to increase the connectivity to $k+1$ we refer to~\cite{N18} and references therein. 
For the SND problem, we refer to the survey by Kortsarz and Nutov~\cite{KN07} and a more recent paper by Cheriyan and V\'egh~\cite{CV14}.

%
%Demaine and Zadimoghaddam~\cite{DZ10} studied the problem of minimizing the eccentricity of a graph by adding a limited number of new edges. A 4-approximation algorithm is introduced and it is proven that the problem is $\mathit{NP}$-hard to be approximated within a factor smaller than $3/2$.
% 
The problem of minimizing the average all-pairs shortest path distance of a graph has been studied by Papagelis~\cite{papagelis2015refining}.
%The author proves that this problem is $\mathit{NP}$-Hard and proposes a path screening technique to select the edges to be added.
The problem of adding a small set of links in order to maximize the centrality of a given vertex in a network  has been addressed for different centrality measures: page-rank~\cite{AN06,OV14}, eccentricity~\cite{DZ10}, average distance~\cite{MT09}, harmonic and betweenness centrality~\cite{BCDMSV18,CDSV15,CDSV16}, and coverage centrality~\cite{DOS19}.

In this paper we focus on the \emph{Maximum Connectivity Improvement} (MCI) problem, which consists in adding at most $B$ edges to a directed graph $G=(V,E)$ in order to maximize the number of pairs of vertices $(u,v)\in V\times V$ such that $v$ is reachable from $u$ in the augmented graph. 
This problem is a maximization version of the \emph{Strong Connectivity Augmentation} (SCA) problem which asks to add a minimum number of edges to a directed graph in order to make the resulting augmented graph strongly connected~\cite{eswaran1976augmentation}. 
While the SCA problem can be solved in linear time~\cite{eswaran1976augmentation}, the MCI problem is $\mathit{NP}$-hard even in restricted case of Directed Acyclic Graphs (DAGs) with only one source vertex or only one sink vertex~\cite{coroAlgosensor2018}. 

%Motivated by applications in social networks~\cite{DSV17,papagelis2015refining}, t
The MCI problem has been defined in~\cite{coroAlgosensor2018}, where it has been studied from an approximation point of view. The authors show that the problem is $\mathit{NP}$-hard to approximate within some constant factor and propose an algorithm that matches this factor in the case of DAGs with only one sink or only one source. Moreover, they give a polynomial time exact algorithm for the case of trees with a single source or a single sink.

%In this paper we study the MCI problem on DAGs from a parameterized complexity perspective and give several hardness results and FPT algorithms.

%\subsubsection*{Our results}
%In this paper we give a broad characterization of the MCI problem on DAGs from a parameterized complexity point of view.
In this paper we address the MCI problem on DAGs from a parameterized complexity point of view by using many natural parameters.

%%%%%
In parameterized complexity each problem instance comes with a parameter $k$. 
A parameterized problem that can be solved exactly in $\mathcal{O}(f(k)n^c)$ time is said to be \emph{Fixed-Parameter Tractable} (FPT). 
Above FPT, there exists a hierarchy of complexity classes, known as the W-hierarchy.
Just as $\mathit{NP}$-hardness is used as evidence that a problem is probably not polynomial time solvable, showing that a parameterized problem is hard for one of these classes gives evidence to the belief that the problem is unlikely to be fixed parameter tractable. 
The principal analogue of the classical intractability class $\mathit{NP}$ is $\mathit{W}[1]$. 
In particular, this means that an FPT algorithm for any $\mathit{W}[1]$-hard problem would yield an $\mathcal{O}(f(k)n^c)$ time algorithm for every problem in the class $\mathit{W}[1]$.

In this paper, we first prove, by modifying a reduction given in~\cite{coroAlgosensor2018}, that it is hard to find a FPT algorithm for MCI on DAGs with respect to several parameters. 
In detail, we show that MCI is $W[2]$-hard with respect to the number of added edges $B$ and it remains $\mathit{NP}$-hard if one of the following parameters are constant numbers: the sum of maximum in-degree and out-degree $\Delta$, the number of sources $|S|$ (or the number of sinks $|T|$) plus the number of isolated vertices $|Q|$, the difference between the size of the minimum vertex cover $\tau$ and that of the largest independent set $\nu$ in the underlying undirected graph. Similar conclusions hold for the chromatic number ($\chi$) and the size of the largest clique ($\omega$).

On the positive side, we show that the problem is FPT with respect to other (somewhat complementary) parameters, namely: the value of an optimal solution $OPT$, $|S|+|T|$, $|V|-B$, $|V|^2 - OPT$, $\nu$, and $\tau$. 
While it is easy to see that simple brute-force algorithms solve MCI in FPT time with respect to $OPT$, $|S|+|T|$, and $|V|^2 - OPT$, the algorithms for parameters $|V|-B$, $\nu$, and $\tau$ require further arguments. 
Our main results consist in algorithms that run in time $2^{O((|V|-B)\log (|V|-B))}|V|^{O(1)}$ and $2^{2^{O(\nu)}}|V|^{O(1)}$. %, and  $2^{2^{O(\tau)}}|V|^{O(1)}$. 
This latter algorithm implies also an FPT algorithm for parameter $\tau$, since $\tau\geq \nu$ in any graph.

Our results are summarized in Table~\ref{tbl:results}.
It is worth to observe that the MCI problem exhibits different complexity if parameterized with respect to a parameter or to its complement. 
For example, MCI is $W[2]$-hard with respect to $B$ but it is FPT with respect to $|V|-B$.~\footnote{One can assume that $|V|$ is an upper bound for $B$ as otherwise the graph can be made strongly connected (see Theorem~\ref{thm:Tarjan}).} 
Similarly, it is solvable in FTP time with respect to both $\tau$ and $\nu$, while it is $\mathit{NP}$-hard when $\tau-\nu$ is a constant (note that in any graph $\tau-\nu\leq \nu\leq \tau$). 
In contrast, MCI is FPT with respect to both $OPT$ and $|V|^2 - OPT$ ($|V|^2$ is an upper bound to $OPT$).
Finally, if at most one between $|S|$ and $|T|$ is a constant, then the problem remains $\mathit{NP}$-hard, while it is FPT with respect to $|S|+|T|$. 

%\paragraph*{Further related work}

\begin{table}[t]
\centering 
\begin{tabular}{cc||cc}
\multicolumn{1}{c}{Parameter} & Result             & \multicolumn{1}{c}{Parameter} & Result             \\ \hline
$OPT$                          & FPT                & $|V| - B$                        & FPT                \\
$B$                            & $\mathit{W}[2]$-hard        & $|V|^2 - OPT$                    & FPT                \\
$\Delta$                       & $\mathit{NP}$-hard & $\nu$, $\tau$                  & FPT                \\
$|S|+|Q|$, $|T|+|Q|$               & $\mathit{NP}$-hard & $\tau - \nu$                   & $\mathit{NP}$-hard \\
$|S| + |T|$, $\max\{|S|,|T|\}$                   & FPT                &        $\chi,\omega$                        &       $\mathit{NP}$-hard            
\end{tabular}
\caption{Results in this paper.}
\label{tbl:results}
\end{table}

%\begin{table}[ht]
%\begin{tabular}{l|llllllllll}
%Parameter & $OPT$ &  $B$ & $\Delta$ & $|S|+q$, $|T|+q$ & $|S| + |T|$ & $V - B$ & $V^2 - OPT$     & $\nu$, $\tau$ & $\tau - \nu$ & $\chi$\\
%\hline
%Result & FPT & $\mathit{W}[2]$-hard & $NP$-hard & $NP$-hard & FPT &FPT &FPT &FPT &$NP$-hard&$NP$-hard
%\end{tabular}
%\caption{Results in this paper.}
%\label{tbl:results}
%\end{table}

%%%%%

\section{Problem statement and preliminaries}

Let $G = (V, E)$ be an unweighted DAG (Directed Acyclic Graph). Given two vertices $u,v \in V$, we say that $u$ is reachable from $v$ in $G$ if there is a directed path from $v$ to $u$.\\
Akin to that defined in~\cite{coroAlgosensor2018}, our objective is to augment the graph $G$ by adding a set $N$ of edges of at most size $B$, i.e., $|N| \le B$ and  $B\in \N_{\ge 0}$, that maximizes the connectivity of $G$.
Let $f(G)= \sum_{v \in V} |\{ u \in V : \exists \text{ path from $v$ to $u$ in } G\}|$ and $G(N)=(V,E\cup N)$, we formally define the following optimization problem:~\footnote{We can equivalently define MCI as a decision problem without affecting (up to a poly-logarithmic factor) the complexity of the algorithms given in this paper.} 

\begin{definition}[Maximum Connectivity Improvement (MCI)]
Given a DAG $G=(V, E)$ and a budget $B\in \N_{\ge 0}$.
We want to add a set of edges $N^* \subseteq \Gamma = (V \times V) \setminus E$, with $|N^*| \leq B$, such that $f(G(N))$ is maximized. 
That is
\[
    N^* = \argmax_{ N \subseteq \Gamma: |N|\leq B } f(G(N))
\]
\end{definition}

Given a DAG, a vertex with no incoming edges ans at least one outgoing edge is called a \emph{source}, while a vertex with no outgoing edges and at least one incoming edge is called a \emph{sink}. 
In the remainder of the paper we denote with $S$ the set of sources, with $T$ the set of sinks, and with $Q$ the set of isolated vertices in $G$ (i.e. before adding edges). Let us explicitly note that $Q$ is disjoint with $S$ or $T$. 
Finally, we denote with $OPT$ (or sometimes $OPT(G,B)$ in order to make the graph and the budget explicit) and $ALG$ the value, respectively, of an optimal solution and the solution find by the algorithm we are considering.
 We also point out a result by Tarjan et al.~\cite{eswaran1976augmentation} that we will use frequently:
\begin{theorem}
[\cite{eswaran1976augmentation}]\label{thm:Tarjan} Let $G$ be a non-trivial DAG and let $B$ be a positive integer. Then $G$ can be made strongly connected by adding at most $B$ edges ($OPT(G, B)=|V|^2$), if and only if $B\geq \max\{|S|,|T|\}+|Q|$. In the latter case, these edges can be found in polynomial time.
\end{theorem} We say that a DAG is trivial if it contains one vertex.

Therefore, in the reminder of the paper  we will assume that $B<\max\{|S|, |T|\}+|Q|$, as otherwise MCI can be optimally solved in polynomial time. Sometimes, we will also assume that $\max\{|S|, |T|\} = |S|$, as, in those case, equivalent results can be obtained when $|T|\ge |S|$ by using the same arguments.

%Here we assume that $Q$ is the set of isolated vertices of $G$, $S$ is the set of sources of $G$, and $T$ is the set of sinks of $G$.

%implies that if $B\geq \max\{|S|,|T|\}+|Q|$ then there exists a polynomial time algorithm to transform any DAG $G$ in an unique strongly connected component and thus optimally solve the problem.

Next lemma allows us to focus on solutions that contain only edges connecting sink vertices to source vertices.
\begin{lemma}[\cite{coroAlgosensor2018}]
\label{lemma:leaf-root}
Let $N$ be a solution to the MCI problem, then there exists a solution $N'$ such that $|N|=|N'|$, $f(N)\leq f(N')$, and all edges in $N'$ connect sink vertices to source vertices.
\end{lemma}

The following proposition proved in~\cite{Sasak:2010} allows us to reduce one parameter to the another one when investigating the fixed-parameter tractability.

\begin{proposition}
[\cite{Sasak:2010}]
\label{prop:ReducingParameters}
Let $\Pi$ be an algorithmic problem and let $k_1$ and $k_2$ be two parameters. Assume that there is a (computable) function $h:\N\rightarrow \N$ such that for any instance $I$ of $\Pi$, we have that $k_1(I)\leq h(k_2(I))$. Then if $\Pi$ is FPT with respect to $k_1$, then it is FPT with respect to $k_2$. 
\end{proposition}
 
%%%%%%%%%%%%%%%%%%%%%%%%%%%%%%%%%%%%%%%%%%%%%%%%%%%%
%%%%%%%%%%%%%%%%%%%%%%%%%%%%%%%%%%%%%%%%%%%%%%%%%%%%
%%%%%%%%%%%%%%%%%%%%%%%%%%%%%%%%%%%%%%%%%%%%%%%%%%%%

\section{Some hardness results}

Cor\`o et al. in~\cite{coroAlgosensor2018} proved that MCI is $\mathit{NP}$-complete and $\mathit{NP}$-hard to approximate within a factor greater than $1-\frac{1}{e}$ even in the case of graphs with a single sink vertex or a single source vertex. 
In this section we give an alternative reduction to prove that the problem is $\mathit{NP}$-complete. Our reduction is from \emph{Exact Cover by 3-Sets}. Later, it will help us to derive several properties of MCI.

\begin{theorem}\label{thm:exact3cover}
MCI is $\mathit{NP}$-complete.
\end{theorem}
\begin{proof}
Consider the decision version of MCI in which given a directed graph $G=(V, E)$ and two integers $M, B \in \N_{\geq 0}$, the goal is to find a set of additional edges $N \subseteq (V \times V) \setminus E$ such that $f(N) \geq M$ and $\card{N} = B$.
The problem is in $\mathit{NP}$ since it can be checked in polynomial time if a set of edges $N$ is such that  $f(N) \geq M$ and $\card{N} = B$.

We reduce from the \emph{Exact Cover by 3-Sets} (X3C) problem which is known to be $\mathit{NP}$-hard even if no element occurs in more than three subsets~\cite[Problem SP2, page 221]{GJ79}.

Consider an instance of the X3C problem $I_{X3C} = (X, Y, q)$ defined by a collections of $3q$ elements, $X = \{x_1, \ldots,  x_{3q}\}$, and $m$ subsets, $Y = \{y_1, \ldots,  y_m\}$, with $y_i \subseteq X$ and $|y_i| = 3$ for all $i$. 
The problem is to decide whether there exist $q$ subsets whose union is equal to $X$.

We define a corresponding instance $I_{MCI}=(G, M, B)$ of MCI as follows.
Create a graph $G=(V, E)$, where $V = \{v_{x_j}~|~x_j\in X\} \cup \{v_{y_i}~|~y_i\in Y\} \cup \{v\} \cup V'$ 
and $E = \{ (v_{y_i}, v_{x_j})~|~x_j \in y_i \} \cup E'$ and $V', E'$ form the following structure. We connect every pair of vertices $v_{x_j}, v_{x_{j+1}}$ to a new vertex $v_i$, i.e., we add the edges $(v_{x_j}, v_i)$ and $(v_{x_{j+1}}, v_i)$.
We, then, repeat this process with the new vertices $v_i$ just created in order to form a binary tree rooted in the vertex $v$ with all the edges toward $v$.
Note that with this process we create $3q-2$ vertices (excluding the root $v$).
% Finally, every two vertices $v_{x_j}$ and $v_{x_{j+1}}$ we add an edge from these two vertices to one of the leaves of the tree.
It is easy to see that now we have bounded the in-degree and the out-degree of $G$, as problem X3C remains $\mathit{NP}$-hard if each element occurs in at most three subsets. See Fig.~\ref{fig:x3c-reduction} for an example. Then we set $B=q$ and $M = (B + 6q-1)^2 + (m-B)(B+6 q ) = (7q-1)^2 + 7q(m-q)$.

By Lemma~\ref{lemma:leaf-root}, we can assume that any solution $N$ of MCI contains only edges $(v, v_{y_i})$ for some $y_i \in Y$. 
In fact, $v$ is the only sink vertex and $v_{y_i}$ are the only source vertices.
Assume that there exists a set cover $Y'$, then we define a solution $N$ to the MCI instance as $N = \{ (v, v_{y_i})~|~y_i \in Y' \}$. 
It is easy to show that $f(N)\geq M$ and $|N|=q=B$. 
Indeed, all the vertices in $G$ can reach: vertex $v$, all the vertices $v_{x_j}$ (since $Y'$ is a set cover) and the vertices $v_i$ that form the tree that are $3q - 2$, and all the vertices $v_{y_i}$ such that $y_i\in Y'$. 
Moreover, each vertex $v_{y_i}$ such that $y_i\not\in Y'$ can reach itself. 
Therefore there are at least $B +6 q -1$ vertices that are able to reach $B + 6 q -1$ vertices and $m-B$ that reach themselves and $B+6 q -1$ other vertices. Hence, $f(N)\ge M$.
On the other hand, assume that there exists a solution for MCI then $N$ is in the form $\{ (v, v_{y_i})~|~y_i \in Y \}$ and we define a solution for the set cover as $F' = \{y_i~|~(v, v_{y_i}) \in N\}$.
We show that $F'$ is a set cover. 
By contradiction, if we assume that $F'$ is not a set cover and it cover only $3q-k$ elements of $X$ ($k\geq 1$), then $f(N) \leq (B+6q-1-k)^2 + (m-B+k)\cdot (B+6q-k)+(m-q)\cdot k$. Now it is matter of direct verification that this expression is less than $M$. This contradicts the choice of $N$ as a solution for MCI. We observe that in the reduction the graph is bipartite and that there is only one sink vertex $v$ and multiple sources. It is easy to see that the same result holds in graphs with only one source and multiple sinks, by using the transpose graph of $G$ and the same values of $B$ and $M$.
\end{proof}

\begin{figure}[t]
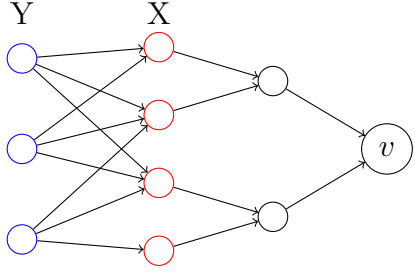

    \centering
    \includestandalone[mode=buildnew, scale=1]{img/tikz/x3c-to-mci}
    \caption{Example of reduction from X3C to MCI used in Theorem~\ref{thm:exact3cover}}
    \label{fig:x3c-reduction}
\end{figure}

\begin{corollary}
MCI is $\mathit{W}[2]$-hard with respect to the budget $B$.
\end{corollary}
This result directly follows either from the reduction in~\cite{coroAlgosensor2018} from Set Cover, that is known to be $W[2]$-hard~\cite{cesati2006compendium}, or from the reduction given in the proof of Theorem~\ref{thm:exact3cover}.

\begin{corollary}
MCI is not FPT with respect to $|S|+|Q|$ and $|T|+|Q|$, unless $\mathit{P=NP}$.
\end{corollary}
\begin{proof}
%The reduction given in~\cite{cesati2006compendium} implies that the problem is $\mathit{W}[2]$-hard with respect to $|S|+|Q|$ and $|T|+|Q|$. 
This just follows from the observation that the number of sources or sinks is one in the reduction, and there are no isolated vertices.
\end{proof}

\begin{corollary}
MCI is not FPT with respect to the chromatic number $\chi$, the size of the largest clique $\omega$ of $G$, and  the difference between the size of the minimum vertex cover $\tau$ and that of the largest independent set $\nu$, unless $\mathit{P=NP}$.
\end{corollary}
\begin{proof}
%The reduction given in~\cite{cesati2006compendium} implies that the problem is $\mathit{W}[2]$-hard with respect to $|S|+|Q|$ and $|T|+|Q|$. 
This just follows by observing that the graphs obtained in the reduction from Theorem~\ref{thm:exact3cover} have bounded $\chi$ and $\omega$. Moreover these graphs are bipartite, hence $\tau(G)=\nu(G)$ and this implies that if $P\neq NP$, we have that the problem is not FPT with respect to $\tau(G)-\nu(G)$.
\end{proof}

\begin{corollary}
MCI is not FPT with respect to $\Delta =\max_{z\in V} (d^-(z)+d^+(z))$, where $d^-(z), d^+(z)$ are, respectively, the maximum in-degree and out-degree in the graph, unless $\mathit{P=NP}$.
\end{corollary}
\begin{proof}
It follows directly from Theorem~\ref{thm:exact3cover}, in fact it is easy to see that in the reduction we have bounded the in-degree and the out-degree of $G$. 
This implies that it cannot exists an FPT with parameter $\Delta$, otherwise, we would have a polynomial algorithm to solve X3C.
\end{proof}

%%%%%%%%%%%%%%%%%%%%%%%%%%%%%%%%%%%%%%%%%%%%%%%%%%%%
%%%%%%%%%%%%%%%%%%%%%%%%%%%%%%%%%%%%%%%%%%%%%%%%%%%%
%%%%%%%%%%%%%%%%%%%%%%%%%%%%%%%%%%%%%%%%%%%%%%%%%%%%

\section{Some FPT results}

In this section, we present our first results on FPT of MCI. Throughout the paper we will use the following propositions.

\begin{proposition}\label{prop:PathProp} 
There is an optimal solution $N^*$,  such that
\begin{enumerate}[(a)]
    \item each isolated vertex of $G$ has in-degree and out-degree at most one in $G(N^*)$, 
    
    \item the isolated vertices of $G$ induce a directed path plus some isolated vertices in $G(N^*)$.
\end{enumerate}
\end{proposition}

\begin{proof}
\begin{enumerate}[(a)]
\item Let $N$ be an optimal solution. Assume that a vertex $v\in Q$ is incident to two edges leaving it in $G(N)$.
Let the neighbors of $v$ be $u$ and $w$. 
Now, if we remove the edge $(v, w)$ and add $(u, w)$, we will still have that the vertex $w$ is reachable from $v$. 
Thus, we can join all these vertices in $|Q|$ with a path and have at least the same reachability as before. We can repeat these arguments in the case that a vertex in $Q$ has more that two outgoing edges in $G(N)$.
Similarly, one can show that the maximum in-degree is at most one. 

\item Thanks to (a), the isolated vertices induce vertex-disjoint directed paths. 
Now, we claim that at most one of these paths can contain at least one edge, others must be isolated vertices. 
On the opposite assumption, choose a directed path $P$ containing the maximum number of edges.
Let the source of $P$ be $u$ and the sink of $P$ be $v$. 
Let us consider any other path $P'$ and let $x$ and $y$ be the source and the sink of $P'$, respectively. Moreover, thanks to (a) there is at most one edge incoming $x$ and one edge outgoing $y$, let $(z,x)$ and $(y,w)$ be the (possible) edges incident to $x$ and $y$, respectively. If edges $(z,u)$ and $(v,w)$ do not belong to the current solution, we can replace edges $(z,x)$ and $(y,w)$ with $(z,u)$ and $(v,w)$ without decreasing the  value of the objective function. Otherwise, if $(z,u)$ is already in the current solution, we merge paths $P$ and $P'$ by using the same number of edges as before: we replace edge $(z,u)$ with $(y,u)$ and $(y,w)$ outgoing $y$ with $(v,w)$. Also in this case the value of the objective function is not decreased. If $(v,w)$ is already in the current solution, we can merge $P$ and $P'$ in a similar way. If both $(z,u)$ and $(v,w)$ are already in the current solution we obtain a solution that does not decreases the objective function and uses one unit of budget less. 
If $x$ has no incoming edge, we can remove $x$ from $P'$ and add it to $P$ without increasing the number of edges or decreasing the objective function. We can do a similar operation if $y$ has no outgoing edges.

By repeating these arguments we end with a single path of vertices in $Q$ in which the end-points can have multiple outgoing or incoming edges. We apply again the arguments in (a) to have at most one incoming and one outgoing edge.
\end{enumerate}
%Now, replace any edge $(x, y)$ entering the source $y$ of other paths with $(x, u)$. 
%Similarly, any edge $(z, w)$ connecting the sink $z$ to a vertex $w$ can be replaced with the edge $(v, w)$. 
%Observe that we still have an optimal solution, i.e., the value of the objective function is not decreased.
\end{proof}

Observe that the initial graph $G$ is a DAG. Thus, all of its strongly connected components contain just one vertex. We will call such strongly connected components trivial.
\begin{proposition}
\label{prop:OneNonTrivalSCC} There is an optimal way of adding $B$ new edges, such that the resulting graph contains at most one non-trivial strongly connected component.
\end{proposition}

\begin{proof} Assume that after adding $B$ edges, the resulting graph contains two non-trivial strongly connected components $C_1$ and $C_2$. Since all components in the original graph were trivial, we have that there are edges $(v_1,u_1)$ in $C_1$ and $(v_2, u_2)$ in $C_2$ that were among these new added $B$ edges. Since $C_1$ and $C_2$ are strongly connected, there are paths $P_1$ in $C_1$ that connects $u_1$ to $v_1$ and $P_2$ in $C_2$ that connects $u_2$ to $v_2$. Consider the way of adding edges, which is obtained from the previous one by replacing $(v_1,u_1)$ and $(v_2, u_2)$ with $(v_1,u_2)$ and $(v_2, u_1)$. Observe that the number of reachable pairs has not decreased. Hence the resulting way of adding edges is optimal. However, this operation has decreased the number of non-trivial strongly connected components.
\end{proof}

\begin{theorem}\label{thm:fpt-s+t}
MCI is FPT with respect to $|S| + |T|$.
\end{theorem}

\begin{proof}%[proof Theorem~\ref{thm:fpt-s+t}]
We start with the case when $Q=\emptyset$. Given an instance of MCI with $G=(V,E)$ and budget $B$, we can always assume that $B \leq \max\{ |S|, |T|\}$ (Theorem~\ref{thm:Tarjan}).
Note that the overall number of possible solutions (as a set of edges) that we can have is $\binom{{|S| |T|}}{B}$ (see Lemma~\ref{lemma:leaf-root}).
%
%$\binom{n^2}{B}$.
%Therefore, a simple algorithm can try all of these sets in time $\mathcal{O}^*(n^{2 B}) \leq \mathcal{O}^*(n^{2n}) \leq \mathcal{O}^*(OPT^{2 OPT})$. 
%Hence, MCI is FPT respect to $OPT$.
%Using the same reason we can write the number of possible solutions as $\binom{{|S| |T|}}{B}$, 
Thus, for the running-time of the trivial algorithm that tries to all possible ways of joining sinks to sources using $B$ edges, we will have the following bound:
\[\mathcal{O}^*((|S| |T|)^{B}) \leq \mathcal{O}^*((|S|+ |T|)^{2 (|S|+ |T|)}),\]
as $B \leq \max\{ |S|, |T|\}\leq |S|+|T|$. 
%Observe that this reasoning works only when $|Q|=0$, i.e., there are no isolated vertices.

Now, we complete the case $|Q|\neq 0$ by using Proposition~\ref{prop:PathProp} in the following way. Assume that in an optimal solution $OPT$, the isolated vertices induce a path of length $k$. 
We have $0\leq k\leq |Q|-1$. 
Let this path be $P$. 
Observe the rest of edges (we have $B-k$ of them) will either join a vertex of $T$ to that of $S$, or a vertex of $T$ to the beginning of $P$, or the end-point of $P$ to a vertex of $S$. 
Thus, we can view $P$ as one big vertex and remove the other isolated vertices. 

Just observe that each vertex of $T$ can join to $S$, to $P$ or the remaining $B-k$ isolated vertices.
Thus each vertex of $T$ has $|S|+B-k+1$ choices. 
Similarly, each vertex of $S$ has at most $|T|+B-k+1$ choices. 
Thus, the overall number of choices is bounded by $|T|\cdot (|S|+B-k+1)+|S|\cdot (|T|+B-k+1)$.
Hence we have at most $\binom{|T|\cdot (|S|+B-k+1)+|S|\cdot (|T|+B-k+1)}{B-k}$ possibilities of adding new edges. 
Clearly, we can assume that $B-k\leq |S|+|T|+1$, as otherwise, $G$ can be made to a strongly connected component in polynomial time (Theorem~\ref{thm:Tarjan}).
Thus
\begin{align*}
    &\binom{|T|\cdot (|S|+B-k+1)+|S|\cdot (|T|+B-k+1)}{B-k}
    \\
    &\quad= \mathcal{O}\left(\left[|T|\cdot (|S|+B-k+1)+|S|\cdot (|T|+B-k+1)\right]^{B-k}\right)
    %\\
    %&\quad\leq \mathcal{O}\left(\left[|T|\cdot (|S|+B-k+1)+|S|\cdot (|T|+B-k+1))\right]^{|S|+|T|+1}\right).
\end{align*} Since $B-k$ is bounded in terms of $|S|+|T|$, we have that this algorithm has running-time bounded in terms of $|S|+|T|$. Thus, if we knew the value of $k$ in the optimal solution, the running time of the algorithm will be a function of $|S|+|T|$ times some polynomial in the input size. 

To complete the proof, we do guessing of $k$, that is, we try all possibilities of the values of $k$, that is, $k=0,1,\ldots,|Q|-1$. 
%Observe that the running time of this algorithm will be
%\[
%    \mathcal{O}(|Q|\cdot f(|S|+|T|)\cdot poly(size))=\mathcal{O}( f(|S|+|T|)\cdot poly(size)),
%\]
Since $|Q|\leq |V|$, this increases the running-time of the algorithm by only of a polynomial factor.
\end{proof}

We observe that Theorem~\ref{thm:fpt-s+t} and Proposition~\ref{prop:ReducingParameters} imply that MCI is FPT also with respect to parameter $\max\{|S|,|T|\}$, as $|S|+|T| \leq 2 \max\{|S|,|T|\}$.

\begin{corollary}\label{cor:V2minusOPT} 
MCI is FPT with respect to $|V|^2-OPT$.
\end{corollary}
\begin{proof} 
We assume that $|S|=\max\{|S|, |T|\}$ since the following arguments can be used if $|T|\ge |S|$. Moreover, we assume that $B<\max\{|S|, |T|\}+|Q|=|S|+|Q|$, as otherwise we can make $G$ strongly connected in polynomial time. 
Now, if we add $B$ new edges, we will still have $|S|+|Q|-B$ vertices of $G$  with no incoming edges that remains with no incoming edges the addition of the new edges. Thus no other vertex can reach them. 
Hence
\[
    OPT\leq |V|^2-(|V|-1)\cdot (|S|+|Q|-B) \qquad\text{ or }\qquad (|V|-1)\cdot (|S|+|Q|-B) \leq |V|^2-OPT.
\]
Hence by Proposition~\ref{prop:ReducingParameters}, if we parameterize the problem with respect to $(|V|-1)\cdot (|S|+|Q|-B)$, we will have the result with respect to $|V|^2-OPT$. Consider the trivial brute-force algorithm running in time $\mathcal{O}^*(|V|^{2|V|})$. Observe that
\[
    |V|\leq 2\cdot (|V|-1)\leq  2\cdot (|V|-1)\cdot (|S|+|Q|-B).
\]
Thus, the trivial algorithm is an FPT algorithm with respect to $(|V|-1)\cdot (|S|+|Q|-B)$.
\end{proof}

\begin{corollary}\label{cor:OPT} 
MCI is FPT with respect to $OPT$.
\end{corollary}
\begin{proof} 
Observe that for any DAG $G$, we have $OPT\geq |V|$, since each vertex reaches itself. Thus,
\[
    |V|^2-OPT\leq OPT^2-OPT=h(OPT).
\]
Hence the result follows from Corollary~\ref{cor:V2minusOPT} and Proposition~\ref{prop:ReducingParameters}.
\end{proof}

%%%%%%%%%%%%%%%%%%%%%%%%%%%%%%%%%%%%%%%%%%%%%%%%%%%%
%%%%%%%%%%%%%%%%%%%%%%%%%%%%%%%%%%%%%%%%%%%%%%%%%%%%
%%%%%%%%%%%%%%%%%%%%%%%%%%%%%%%%%%%%%%%%%%%%%%%%%%%%

\section{FPT with respect to \texorpdfstring{$|V|-B$}{|V| - B}}

In this section, we parameterize our problem with respect to $|V|-B$. We will assume that $\max\{|S|,|T|\} = |S|$, as the same arguments can be used when $|T|\geq |S|$. Our proof will rely on the following lemma.

\begin{lemma}\label{lem:BVminusS} 
Given a DAG $G$, assume $Q=\emptyset$, $|S|\geq |T|$, and let $B \geq |V\setminus S|$. 
Then we can find the optimal solution in polynomial time.
\end{lemma}
\begin{proof} 
We can assume that $B \leq \max\{|S|, |T|\} + |Q|=|S|$, as otherwise the problem is polynomial time solvable (Theorem~\ref{thm:Tarjan}). Thus, $|V\setminus S|\leq |S|$.
As in Corollary~\ref{cor:V2minusOPT} we have the following upper bound on $OPT$:
\[
    OPT\leq |V|^2-(|S|-B)\cdot (|V|-1).
\]

Consider the following algorithm: Take a representative source from $S$ for each vertex of $V\setminus S$. A source $s\in S$ is representative for a vertex $z$, if in $G$ there is a directed path that connects $s$ to $z$. Put these $|V\setminus S|$ sources into a subset of $S$ of cardinality $|B|$. With $B$ edges we can create a strongly connected component with this subset and $V\setminus S$ (Theorem~\ref{thm:Tarjan}). Thus, we have $(B+|V\setminus S|)^2$ pairs in the strongly connected component. Moreover, the remaining $|S|-B$ sources can reach the vertices of the strongly connected component and themselves. Hence we have 
\begin{align*}
    ALG &\geq (B+|V\setminus S|)^2+(|S|-B)\cdot (1+B+|V\setminus S|)
    \\
    &=|S|-B+(B+|V|-|S|)|V|=|V|^2-(|V|-1)(|S|-B).
\end{align*}
Taking into account the upper bound for $OPT$, we have that this algorithm find the optimum. The proof is complete. 
\end{proof}

\begin{theorem}\label{thm:|V|minusBFPT} 
MCI is FPT with respect to $|V|-B$.
\end{theorem}
\begin{proof} First assume that $Q=\emptyset$, that is, there are no isolated vertices in $G$. Note that if $|V|-B\geq \frac{|V|}{2}$, then $|V|$ is bounded in terms of $|V|-B$, hence the trivial algorithm will be an FPT algorithm in terms of $|V|-B$. 
On the other hand, let us assume that $|V|-B\leq \frac{|V|}{2}$. 
Then $B\geq \frac{|V|}{2}$. 
Since $B \leq \max\{|S|, |T|\} + |Q|=|S|$, then, 
\[
    \frac{|V|}{2} \leq B \leq |S| \qquad\text{ or }\qquad |V\setminus S|=|V|-|S|\leq \frac{|V|}{2} \leq B.
\]
Thus, $B\geq |V\setminus S|$. 
Hence we can solve this case in polynomial time thanks to Lemma~\ref{lem:BVminusS}.

Now, assume that $Q\neq \emptyset$. By Propositions~\ref{prop:PathProp} and~\ref{prop:OneNonTrivalSCC}, we have that the isolated vertices induce a path of length $k$ ($0\leq k\leq |Q|-1$). Moreover, the end-vertices of this path are joined to at most one vertex outside it. Consider the graph $G'=G\setminus Q$, and let $B'=B-(k+1)$. Observe that the added edges must form an optimal solution for this new instance. This allows us, to finish the proof by guessing $k$, that is, we try all possible values of $k=0,1,\ldots,|Q|-1$. For each fixed $k$, we observe that $G'$ contains no isolated vertex, hence we can solve this instance with the approach outlined above. By taking the one with maximum optimal number of pairs we will get an optimal solution for $G$. Observe that since there are $|Q|$ possible values of $k$, and for each $k$, $|V'|-B'\leq |V|-B$, we will have that the running-time of the algorithm will be FPT in terms of $|V|-B$.
\end{proof}

%%%%%%%%%%%%%%%%%%%%%%%%%%%%%%%%%%%%%%%%%%%%%%%%%%%%
%%%%%%%%%%%%%%%%%%%%%%%%%%%%%%%%%%%%%%%%%%%%%%%%%%%%
%%%%%%%%%%%%%%%%%%%%%%%%%%%%%%%%%%%%%%%%%%%%%%%%%%%%

\section{FPT with respect to the matching number}

In this section, we parameterize the problem with respect to the size of the largest matching. Recall that a matching in a DAG is a subset of edges, such that no two edges of the subset share a vertex. Let $\nu(G)$ be the matching number of $G$, that is, the size of the largest matching in $G$. Our proof will require the following lemma.

\begin{lemma}\label{lem:IsomorphismClasses} 
Let $H=(X, Y, E)$ be a bipartite graph, such that for each $x_1, x_2\in X$, we have that $N(x_1)\neq N(x_2)$, where $N(x)$ denotes the set of neighbors of the vertex $x$. 
Then:
\[
    |X|\leq \nu(H)+2^{\nu(H)}.
\]
\end{lemma} 
\begin{proof} 
Let $M$ be a maximum matching. Let $X_M$ and $Y_M$ be the set of vertices in $X$ and $Y$, respectively, that are covered by $M$. 
Clearly, $|X_M|=|Y_M|=\nu(H)$. Since $M$ is maximum, no vertex in $X\setminus X_M$ is adjacent to that of in $Y\setminus Y_M$. Thus, the neighbors of $X\setminus X_M$ are in $Y_M$. Since different vertices have different neighbors, we have that 
\[
    |X\setminus X_M|\leq 2^{|Y_M|}.
\]
Hence, 
\[
    |X|\leq |X_M|+2^{|Y_M|}=\nu(H)+2^{\nu(H)}.
\]
The proof is complete.
\end{proof}

\begin{theorem}\label{thm:FPTnu}
MCI is FPT with respect to $\nu(G)$.
\end{theorem}
\begin{proof} Again, we start with the case $Q=\emptyset$. Let $S_0$ and $T_0$ be a minimal subset of sources and sinks, respectively, such that they are representatives of the remaining graph, that is, for each vertex $v\in V$, there is a vertex $s_v\in S_0$ and a vertex $t_v\in T_0$ such that there is a path from $s_v$ to $v$ and a path from $v$ to $t_v$ inn $G$, and $S_0$ and $T_0$ are minimal inclusion-wise.

We claim that $|S_0|, |T_0|\leq \nu(G)$. Let us prove for $S_0$. Since it is minimal, we have that for each source $s\in S_0$, there is a vertex $z_s\in V\setminus S$ such that $s$ is the only vertex of $S_0$ that is connected to $z_s$ with a path. Now, if we consider these paths for all vertices of $S_0$, then clearly they are vertex disjoint. 
Hence, the first edges of these paths form a matching of cardinality $|S_0|$. 
Thus, $|S_0|\leq \nu(G)$. 

We distinguish two cases. First, assume that $B\geq \max\{|S_0|, |T_0|\}$. As in the previous section, we can assume that $B \leq \max\{|S|, |T|\} + |Q|=|S|$. Now, let us assume that $B\leq |T|$. 
In this case, any solution, will leave $|S|-B$ sources of $G$ still sources, and similarly, $|T|-B$ sinks of $G$ will be still sinks after adding the $B$ edges. 
Hence, we have the following upper bound for $OPT$:
\[
    OPT\leq |V|^2-(|S|-B)\cdot (|V|-1)-(|T|-B)\cdot (|V|-1)=|V|^2-(|V|-1)\cdot (|S|+|T|-2B).
\]
Since $B\geq \max\{|S_0|, |T_0|\}$, and $B\leq |T|\leq |S|$, we can add new sources to $S_0$ and sinks to $T_0$, such that we get sets of cardinality $B$. 
Now, with $B$ edges we can make this part a strongly connected component together with all internal vertices (Theorem~\ref{thm:Tarjan}). Observe that outside the strongly connected component, there will be $|S|-B$ sources and $|T|-B$ sinks. Thus, the strongly connected  component will have $|V|-(|S|-B)+(|T|-B)=|V|-|S|-|T|+2B$ vertices. 
Hence we will have $(|V|-|S|-|T|+2B)^2$ pairs in the solution. Moreover, the remaining sources can reach the strongly connected component, and the vertices of the strongly connected component can reach the remaining sinks. 
Thus, in total we will have at least
\begin{align*}
    ALG & \geq (|V|-|S|-|T|+2B)^2+(|S|-B)\cdot (1+|V|-|S|-|T|+2B)
    \\
    &\qquad+(|T|-B) \cdot (1+|V|-|S|-|T|+2B)
    \\
    &= (|V|-|S|-|T|+2B)^2+(1+|V|-|S|-|T|+2B)\cdot(|S|+|T|-2B)
    \\
    &= (|S|+|T|-2B)+|V|\cdot (|V|-|S|-|T|+2B)
    \\
    &=|V|^2-(|V|-1)\cdot (|S|+|T|-2B) = OPT.
\end{align*} 
Thus, we have an optimal number of pairs. 
On the other hand, if $|T|\leq B\leq |S|$, then we can overcome this case similarly. Observe that any solution will leave $|S|-B$ sources of $G$ still sources. 
Hence, we will have the following upper bound for $OPT$:
\[
    OPT\leq |V|^2-(|V|-1)\cdot (|S|-B).
\]
Let us add new sources from $S$ to $S_0$, such that we have exactly $B$ sources in it. 
Clearly, we can create a strongly connected component with these $B$ sources and the remaining part of the graph (except $|S|-|S_0|$ sources) (Theorem~\ref{thm:Tarjan}). 
Moreover, the remaining sources will reach the strongly connected component. 
Hence a polynomial time algorithm finds the optimal solution, in fact
\begin{align*}
    ALG &\geq (|V|-|S|+B)^2+(|S|-B)\cdot (1+|V|-|S|+B)
    \\
    &=|S|-B+(|V|-|S|+B)\cdot |V|=|V|^2-(|V|-1)\cdot (|S|-B) = OPT.
\end{align*}

It remains to consider the case $B\leq \max\{|S_0|, |T_0|\}$. 
Hence $B\leq \nu(G)$. Let us say that two sources are equivalent, if they are adjacent to the same set of vertices.
Observe that this relation is an equivalence relation defined on the set of sources. 
Thus, it partitions $S$ into equivalence classes. We claim that the number of equivalence classes is bounded by some function of $\nu(G)$. Assume that $S_1$,\ldots, $S_l$ are the equivalence classes. Consider a bipartite graph $H=(X, Y, E)$, where $X=\{S_1,\ldots, S_l\}$, $Y=N_G(S)$ and $S_i$ is joined to $y\in Y$, if and only if in $G$ there was a vertex in $S_i$ that was adjacent to $y$. 
Since $S_1,\ldots, S_l$ are pairwise different equivalence classes, we have that $H$ satisfies the conditions of Lemma~\ref{lem:IsomorphismClasses}. 
Thus,
\[
    l=|X|\leq \nu(H)+2^{\nu(H)}\leq \nu(G)+2^{\nu(G)}.
\]
Similarly, one can show that the number of equivalence classes of sinks is bounded in terms of $\nu(G)$. 
Now, let $S_1,\ldots, S_l$ and $T_1,\ldots, T_m$ be the equivalence classes of sources and sinks. Consider all possible partitions of $B$ into the sum of $B_{ij}$s, where $i=1,\ldots, l$ and $j=1,\ldots, m$. 
Intuitively, $B_{ij}$ shows how much budget we spend in connecting sinks in $T_j$ to sources in $S_i$. Since each of $B_{ij}$s can be at most $B$, we have that the total number of partitions is at most $(B+1)^{l\cdot m}$. Since $B$ is bounded by $\nu(G)$, we have that this expression is bounded by some function of $\nu(G)$.

Let us show that the number of different ways of joining $B_{ij}$ edges between $T_j$ and $S_i$ is bounded by $B$, hence by $\nu(G)$. We will estimate the number of different configurations in terms of $B$. Observe that we can have at most $B$ sinks of $T_j$ that will be joined to a source from $S_i$. 
Moreover, the number of sources that will be incident to at least one edge is at most $B$, too. Let us count the number of different (not necessarily non-isomorphic) configurations. 
Let $t$ be any sink from $T_j$ that is joined with at least one of $B$ edges. 
Observe that it can be joined to some subset of $S_i$. 
Hence, the number of possibilities of $t$ is at most $2^B$. Recall that we have at most $B$ vertices in $T_j$. 
Thus, the total number of configurations is bounded by
\[
    \leq 2^B\cdot 2^B\cdot \ldots \cdot 2^B=2^{B^2}.
\]

This implies that the number of non-isomorphic ways of joining $B$ edges is bounded by a function of $\nu(G)$.
Thus, we can consider all of them and find the one maximizing the number of connected pairs. 
Clearly, the running-time of this simple algorithm will be FPT with respect to $\nu(G)$. 

Now, assume that $Q\neq \emptyset$. By Propositions~\ref{prop:PathProp} and~\ref{prop:OneNonTrivalSCC}, we have that the isolated vertices induce a path of length $k$ ($0\leq k\leq |Q|-1$). Moreover, the end-vertices of this path are joined to at most one vertex outside it. Consider the graph $G'=G\setminus Q$, and let $B'=B-(k+1)$. Observe that the added edges must form an optimal solution for this new instance. This allows us, to finish the proof by guessing $k$, that is, we try all possible values of $k=0,1,\ldots,|Q|-1$. For each fixed $k$, we observe that $G'$ contains no isolated vertex, hence we can solve this instance with the approach outlined above. By taking the one with maximum optimal number of pairs we will get an optimal solution for $G$. Observe that since there are $|Q|$ possible values of $k$, and for each $k$, $\nu(G')\leq \nu(G)$, we will have that the running-time of the algorithm will be FPT in terms of $\nu(G)$. The proof is complete.
\end{proof}

Observe that in any graph $G$, we have $\nu(G)\leq \tau(G)$, hence we have that the problem is FPT with respect to $\tau(G)$ as well (Proposition~\ref{prop:ReducingParameters}).

%%%%%%%%%%%%%%%%%%%%%%%%%%%%%%%%%%%%%%%%%%%%%%%%%%%%
%%%%%%%%%%%%%%%%%%%%%%%%%%%%%%%%%%%%%%%%%%%%%%%%%%%%
%%%%%%%%%%%%%%%%%%%%%%%%%%%%%%%%%%%%%%%%%%%%%%%%%%%%

\section{Conclusion and future work}

We addressed the MCI problem from a parameterized tractability viewpoint and show hardness and algorithmic results on different natural parameters. Our results open several research directions. The main open problem regards the case of general directed graphs. A strictly more general case is that in which each vertex is associated with a weight and the objective function is the sum, for each $(u,v)\in V\times V$, of the product between the weight of $u$ and that of $v$, if $v$ is reachable from $u$~\cite{coroAlgosensor2018}. A more restricted open problem is to solve the case in which the graph is a directed tree with more than one source. It is worth to decrease the running time of some of the FPT algorithms in this paper, for example find an algorithm with single exponential running time in $\nu$. Another interesting problem that the parameterization with respect to $\nu$, hence $\tau$, suggests is the following: since our problem is FPT with respect to $\tau$, it is FPT with respect to $|V|-\alpha$ , where $\alpha$ is the maximum number of independent vertices of the graph. Thus, an interesting question is the parameterization with respect to $\alpha$. Finally, we would like to suggest the following question:

\begin{question}
\label{que:MaxSTminusB} Is our problem FPT with respect to $\max\{|S|, |T|\}-B$.
\end{question} It is not hard to see that a positive answer to this question will imply four of our results. That are, the parameterization with respect to $\max\{|S|, |T|\}$, $|V|-B$, $|V|^2-OPT$ and $OPT$.

%%%%%%%%%%%%%%%%%%%%%%%%%%%%%%%%%%%%%%%%%%%%%%%%
%%%%%%%%%%%%%%%% BIBLIOGRAPHY %%%%%%%%%%%%%%%%%%
%%%%%%%%%%%%%%%%%%%%%%%%%%%%%%%%%%%%%%%%%%%%%%%%
\bibliographystyle{elsarticle-num}
% \bibliographystyle{elsarticle-harv}
% \bibliographystyle{elsarticle-num-names}
% \bibliographystyle{model1a-num-names}
% \bibliographystyle{model1b-num-names}
% \bibliographystyle{model1c-num-names}
% \bibliographystyle{model1-num-names}
% \bibliographystyle{model2-names}
% \bibliographystyle{model3a-num-names}
% \bibliographystyle{model3-num-names}
% \bibliographystyle{model4-names}
% \bibliographystyle{model5-names}
% \bibliographystyle{model6-num-names}

%\bibliography{sample}
%\section*{References}

\end{document}

%% file: img/tikz/x3c-to-mci.tex
\begin{tikzpicture}[scale=1.5]

\node[style={draw=black, circle}] (1) at (4,0) {$v$};

\node[style={draw=blue, circle}] (2) at (0.8,0) {};
\node[style={draw=blue, circle}] (3) at (0.8,0.8) {};
\node[style={draw=blue, circle}] (4) at (0.8,-0.8) {};

\node[style={draw=red, circle}] (5) at (2,0.9) {};
\node[style={draw=red, circle}] (6) at (2,0.3) {};
\node[style={draw=red, circle}] (7) at (2,-0.3) {};
\node[style={draw=red, circle}] (8) at (2,-0.9) {};

\node[style={draw=black, circle}] (9) at (3,0.6) {};
\node[style={draw=black, circle}] (10) at (3,-0.6) {};

\draw[->] (2) -- (5) {};
\draw[->] (2) -- (6) {};
\draw[->] (2) -- (7) {};

\draw[->] (3) -- (6) {};
\draw[->] (3) -- (7) {};
\draw[->] (3) -- (5) {};

\draw[->] (4) -- (6) {};
\draw[->] (4) -- (7) {};
\draw[->] (4) -- (8) {};

\draw[->] (5) -- (9) {};
\draw[->] (6) -- (9) {};
\draw[->] (7) -- (10) {};
\draw[->] (8) -- (10) {};

\draw[->] (9) -- (1) {};
\draw[->] (10) -- (1) {};

\node[] at (0.8, 1.2) {Y};

\node[] at (2, 1.2) {X};

\end{tikzpicture}